\documentclass[11pt]{article}

\usepackage{fullpage}
\usepackage{amssymb,amsmath}
\usepackage{epsfig}
\usepackage{palatino}
\usepackage{graphicx}
\usepackage{textcomp}
\usepackage{amsthm}
\usepackage{mdframed}
\usepackage{color, caption}
\usepackage{tocloft}
\usepackage[colorlinks=true,linkcolor=blue,linktoc=all]{hyperref}
\newtheorem{theorem}{Theorem}
\newtheorem{proposition}[theorem]{Proposition}

\newtheorem{definition}[theorem]{Definition}
\newtheorem{algorithm}[theorem]{Algorithm}
\newtheorem{claim}[theorem]{Claim}
\newtheorem{lemma}[theorem]{Lemma}
\newtheorem{conjecture}[theorem]{Conjecture}

\newtheorem{corollary}[theorem]{Corollary}

\newtheorem{remark}[theorem]{Remark}
\newtheorem*{theorem*}{Theorem}
\newtheorem*{lemma*}{Lemma}

\newcommand{\nc}{\newcommand}
\nc{\rnc}{\renewcommand}

\def\ba#1\ea{\begin{align}#1\end{align}}
\def\bas#1\eas{\begin{align*}#1\end{align*}}
\def\bpm#1\epm{\begin{pmatrix}#1\end{pmatrix}}
\nc{\nn}{\nonumber}
\nc{\eq}[1]{(\ref{eq:#1})}
\nc{\eqs}[2]{(\ref{eq:#1}) and (\ref{eq:#2})}

\def\begsub#1#2\endsub{\begin{subequations}\label{eq:#1}\begin{align}#2\end{align}\end{subequations}}
\nc\qand{\qquad\text{and}\qquad}
\nc\mnb[1]{\medskip\noindent{\bf #1}}

\nc\benum{\begin{enumerate}}
\nc\eenum{\end{enumerate}}
\newcommand{\bea}{\begin{eqnarray}}
\newcommand{\eea}{\end{eqnarray}}
\nc\bit{\begin{itemize}}
\nc\eit{\end{itemize}}
\nc{\ot}{\otimes}
\rnc{\L}{\left} 
\nc{\R}{\right}

\DeclareMathOperator*{\E}{\mathbb{E}}

\def\R{\mathbb{R}}

\def\C{\mathbb{C}}

\def\poly{{\rm poly}}
\def\log{{\rm log}}

\def\Per{\rm Per}

\newcommand{\be}{\begin{eqnarray}}
\newcommand{\ee}{\end{eqnarray}}

\newcommand{\twodots}{..}

\newcommand{\ignore}[1]{}

\newcommand{\eps}{\varepsilon}
\renewcommand{\epsilon}{\varepsilon}

\nc{\hin}{h_{\text{in}}}
\nc{\pin}{\partial_{\text{in}}}
\nc{\pell}{\partial_{\ell}}

\newcommand{\nocontentsline}[3]{}
\newcommand{\tocless}[2]{\bgroup\let\addcontentsline=\nocontentsline#1{#2}\egroup}
\makeatletter
\newcommand{\cftsectionprecistoc}[1]{\addtocontents{toc}{%
  {\leftskip \cftsecindent\relax
   \advance\leftskip \cftsecnumwidth\relax
   \rightskip \@tocrmarg\relax
   \textit{#1}\protect\par}}}
\makeatother

\begin{document}

\title{
		\huge Approximating the Permanent of a Random Matrix with Vanishing Mean \\
}
\author{Lior Eldar \thanks{eldar.lior@gmail.com}, Saeed Mehraban 
\thanks{{MIT Computer Science and Artificial Intelligence Laboratory (CSAIL)}, mehraban@mit.edu}}

\date{\today}

\maketitle

\begin{abstract}
The permanent is $\#P$-hard to compute exactly on average for natural random matrices including matrices over finite fields or Gaussian ensembles. 
Should we expect that it remains $\#P$-hard to compute on average if we only care about approximation instead of exact computation?

In this work we take a first step towards resolving this question: 
We present a quasi-polynomial time deterministic algorithm for approximating the permanent of a typical $n\times n$ random matrix with unit variance and vanishing mean $\mu=O(\ln \ln n)^{-1/8}$  
to within inverse polynomial multiplicative error.
Alternatively, one can achieve permanent approximation for matrices with mean $\mu = 1/\poly\log(n)$ in time $2^{O(n^{\eps})}$, for any $\eps>0$.

The proposed algorithm significantly extends the regime of matrices for which efficient approximation of the permanent is known.
This is because unlike
previous algorithms which require a stringent correlation between the signs of the entries of the matrix \cite{Barvinok2016, JSV04} it can tolerate random
ensembles in which this correlation is negligible (albeit non-zero).
Among important special cases we note:
\begin{enumerate}
\item
\textbf{Biased Gaussian}: each entry is a complex Gaussian with unit variance $1$ and mean $\mu$.
\item
\textbf{Biased Bernoulli}: each entry is $-1 + \mu$ with probability $1/2$, and $1$ with probability $1/2$.
\end{enumerate}
These results counter the common intuition that the difficulty of computing the permanent, even approximately, stems merely from our
inability to treat  matrices with many opposing signs.
The Gaussian ensemble approaches the threshold of a conjectured hardness \cite{AA13} of computing the permanent of a zero mean Gaussian matrix.
This conjecture is one of the baseline assumptions of the BosonSampling paradigm that has received vast attention in recent
years in the context of quantum supremacy experiments.

We furthermore show that the permanent of the biased Gaussian ensemble is $\# P$-hard to compute exactly on average. 
To our knowledge, this is the first natural example of a counting problem that becomes easy only when average case and approximation are combined.

On a technical level, our approach stems from a recent approach taken by Barvinok \cite{Barvinok2016,Barvinok2013,Bar16,BS17}
who used
Taylor series approximation of the logarithm of a certain univariate polynomial related to the permanent.
Our main contribution is to introduce an average-case analysis of such related polynomials. 
We complement our approach with a new technique for iteratively computing a Taylor series
approximation of a function that is analytical in the vicinity of a curve in the complex plane. 
This method can be viewed as a computational version of analytic continuation in complex analysis.
\end{abstract}

\tocless

\section{Introduction}

\subsection{Complexity of computing the permanent}

The permanent of an $n\times n$ matrix $A$ is the following degree $n$ polynomial in the entries of $A$:
\be
\Per(A):= \sum_{\sigma\in S_n} \prod_{i=1}^n A_{i,\sigma(i)},
\ee
{where $S_n$ is the symmetric group over $n$ elements}. Its computation has been the subject of intense research \cite{Bar16,AA13,JSV04,Gur05,Val79,LSW00} 
and has been connected to subjects ranging from random sampling of bi-partite matchings \cite{JSV04}
to establishing a so-called ``quantum supremacy'' using linear optical experiments \cite{AA13}.

Much of this interest is centered around the computational complexity of computing the permanent.
The permanent is known to be $\#P$-hard to compute exactly \cite{Val79, A2011} 
and we only know exponential time algorithms for computing the permanent of a general matrix $A$, the fastest of which is the
Ryser formula.

Because of this hardness in the worst case, research has also focused on
computing an approximation for the value of permanent:
for multiplicative approximation,
approximation schemes are known for several special cases of matrices.
Perhaps most prominently is the work
of Jerrum, Sinclair and Vigoda
\cite{JSV04} who showed a randomized polynomial time algorithm to compute a $1 + 1/\poly(n)$ multiplicative approximation
of the permanent for matrices with non-negative entries.
More recently \cite{AGGS17} have shown how to approximate the permanent of a PSD
matrix to simply exponential factor in polynomial time.

Still, if one allows the matrix to have an arbitrary number of negative values it is even $\#P$-hard to compute the sign of the permanent \cite{A2011}, which rules out a multiplicative approximation. This suggests that part of the computational hardness of computing the permanent comes from alternating signs in a matrix. Hence, for general matrices, it seems that efficient approximation of the permanent remains well out of reach
and progress is limited by a an ``interference barrier'' by which 
the positive and negative entries of a matrix may generate an intricate interference pattern that is hard
to approximate.

{Given the apparent difficulty in computing the permanent exactly and approximately one may ask a different question: Is the computation of permanent still hard on a faction of inputs?
It turns out that it is still difficult to compute the permanent even on a small fraction of inputs. {For example, it has been shown \cite {cps99} that the permanent of a matrix over a finite field is $\#P$-hard to compute exactly even on a $1/ \poly(n)$ fraction of such matrices}.
Specifically for the case of complex Gaussian matrices a somewhat weaker statement is known:
It is $\#P$-hard to compute $\Per(A)$ exactly for a Gaussian $A$ w.p. greater than $3/4 + 1/\poly(n)$ \cite{AA13}.  }

Faced with the apparent robustness of the complexity of computing the permanent against both error on a fraction
of inputs and against approximation, Aaronson and Arkhipov \cite{AA13}
designed a beautiful paradigm called ${\rm BosonSampling}$ for demonstrating a so-called quantum supremacy over classical computers.
It hinges upon permanent computation remaining $\#P$-hard even when we simultaneously
allow error on some inputs (see table \ref{table:1}), and allow an approximate value for the inputs for which we do handle.
The intuitive difficulty of approximating the permanent on such a distribution stems  from the same "interference barrier" described above.
Namely,  that quantum computers,
using complex amplitudes encoded in quantum states can "bypass" this barrier naturally, whereas classical computers fall short.

The range of parameters in our result
approaches 
the threshold of this conjectured hardness.
Hence,  it raises the intriguing question
of whether there exists a yet undiscovered phenomenon,
that occurs only for 
zero mean, which prohibits efficient approximation, or whether
that conjecture is false. 
We discuss this further in Section \ref{sec:BS}.
\begin{table}
\captionsetup{font=scriptsize}
\captionsetup{width=13cm}

\centering
\captionsetup{width=.8\linewidth}
\begin{tabular}{ |c|  c|  c| }
\hline
   & worst case & average case \\
   \hline
 exact & $\#P$-hard & $\#P$-hard \\
 \hline
   approximate & $\#P$-hard &? \\
   \hline
\end{tabular}
\caption{The computational complexity of computing the permanent of a complex Gaussian matrix. We know that: Permanent is $\#P$-hard to compute exactly in the worst case or average case. We also know that permanent is $\#P$-hard to compute exactly on average. However, nothing is known about the approximation of permanent in the average case. Our result demonstrates that the permanent can be approximated in quasi-polynomial time if the Gaussian ensemble has non-zero but vanishing mean.}
\label{table:1}
\end{table}

\subsection{A complex function perspective}\label{sec:barvinok}

Recent works of Barvinok \cite{Barvinok2013,Bar16,Barvinok2016, BS17, Barvinok2018}
have outlined a new approach to computing the permanent,
and in fact a large class of high-degree polynomials in matrices
such as the partition function
\cite{PATEL2017,LSS17}.
His approach, stated here for the permanent of an $n\times n$ matrix $A$, is quite intuitive:
instead of trying to compute $\Per(A)$ directly, one computes an additive approximation of $\ln(\Per(A))$
and then exponentiates the result.
Let $J$ be the all ones $n\times n$ matrix.
The approximation of $\ln(\Per(A))$ is then computed as a Taylor series approximation of the complex log
of the univariate polynomial $g_A(z) = \Per(J \cdot (1-z) + z \cdot A)$ around the point $z=0$, and evaluated at point $z=1$.
The crucial point is that one can compute the lowest $m$ derivatives of $f$ at point $z=0$ relatively
efficiently, i.e., in time $n^{O(m)}$ 
since they correspond essentially to a linear combination of the permanents of sub-matrices
of $A$ of size at most $m\times m$.
The additive approximation error of the Taylor series expansion decays exponentially fast in $m$ so choosing $m = O(\ln(n))$
implies an algorithm for a multiplicative error of $1 + 1/\poly(n)$ that runs in time at most ${n\choose m} = 2^{O(\ln^2(n))}$.

In order to apply this Taylor series technique there is a significant limitation that must be carefully addressed: The Taylor approximation
of $\ln(g(z))$
about a point $z = x_0$ is valid only in a disk around $x_0$ that contains no poles of $f(z)$ (the roots of $g(z)$). 
Therefore, this approach inherently requires knowledge about the location of the roots
of the univariate polynomial used for interpolating the value of the permanent from
an easy-to-compute matrix to the target matrix $A$.

Using combinatorial arguments Barvinok characterized the location of the roots of $g_A(z)$ for certain
complex matrices: For example those that that satisfy $\max_{i,j} |A_{i,j} -1| \leq 0.5$ \cite{Barvinok2016}
and diagonally dominant matrices \cite{Barvinok2018}.
This then implied quasi-polynomial time algorithms for these classes of matrices.

Hence, his approach is the first to break the "sign barrier'' for approximating the permanent - i.e.
the ability to approximate the permanent for matrices that have many entries with opposing signs,
thus extending our ability to compute the permanent beyond matrices with non-negative entries
for which the  algorithm by Jerrum, Sinclair and Vigoda \cite{JSV04} is famously known for.
Still, this divergence from non-negative entry matrices was only quite mild: The entries of the matrix types that the algorithm handles
are still highly correlated in a specific direction and so have a very negligible interference pattern.

That said, this approach opens up a wide range of possibilities for computing these otherwise
intractable polynomials: Instead of thinking about them as combinatorial objects,
one can ask a completely different set of questions:
what can be said about the location of the roots of $g_A(z)$?
can one detour around these roots in order to reach ``interesting''
points $z$ where the value of $g_A(z)$ is non-trivial?
Yet another set of questions can then be phrased for a ``random ensemble'' of matrices $A$:
what is the behavior of the roots of $g_A(z)$ for typical $A$ and can they be detoured ``on average'', in an efficient way?
Answering such questions analytically, and subsequently efficiently as an algorithm is the focus of our work. 

\subsection{Main results}

Consider a random matrix where each entry is sampled independently from a distribution of complex valued random variables with mean $0$ and variance $1$. We refer to such matrix a random matrix with mean $0$ and variance $1$. 
An example of such a matrix is a Bernoulli or Gaussian matrix.

\begin{theorem*}[Informal statement of Theorem \ref{thm:main}]
Let $n$ be some sufficiently large integer and $\mu =1/ \poly\log \log n$. 
Let $A$ be an $n \times n$ random matrix with mean $\mu$ and variance $1$. 
There is a deterministic quasi-polynomial time algorithm that for $1-o(1)$ fraction of random matrices $A$ outputs a number that is within inverse polynomial relative error of the permanent $\Per (A)$.
\end{theorem*}

We note that one can also achieve a mean value parameter of $\mu = 1/\poly\log n$ with a run time that is strictly faster than $2^{O(n^{\eps})}$ for any $\eps>0$.
See Remark \ref{remark:subexp}

One can ask whether perhaps allowing a non-vanishing mean devoids the permanent of its average-case hardness.
Thus we show a complementary claim whose proof appears in \cite{EM17}
\begin{theorem}
[Average-case hardness for the permanent of a nonzero mean Gaussian] Let $\mu > 0$ and let $\mathcal{O}$ be the oracle that for any $\mu' \geq \mu$ computes the permanent of $7/8 +1/\poly(n)$ fraction of matrices from the ensemble ${\cal N}^{n \times n}(\mu', 1, \C)$ exactly. Then $P^{\mathcal{O}} = P^{\# P}$.
\label{thm:exacthardness}
\end{theorem}
Hence our results establish a natural ensemble of matrices for which permanent approximation is efficient on average even though the exact computation is not.
This should be contrasted with the central BosonSampling assumption which is that permanent approximation remains hard for $0$-mean.

\subsection{Roots of random interpolating polynomials}

In this work we consider the Taylor series technique of Barvinok (see Section \ref{sec:barvinok}) 
in a random setting: We ask - given an ensemble of random matrices $A$
what can be said about the typical location of roots of some interpolating polynomial related to the permanent of $A$, say $g_A(z) = \Per((1-z) J + z A)$?

Notably, this question completely changes the flavor of the analysis: Instead of considering families of matrices with some fixed property (say diagonal dominance)
and analyzing them combinatorially, we consider random matrices, and then treat $g_A(z)$ as a random polynomial
which allows us to bring to bear techniques from analysis of random polynomials.

First, to simplify matters we consider instead the polynomial
\be
g_A(z) = \Per(J + z A)
\ee
and observe that for any non-zero $z$ if $A$ is a random matrix with mean $0$, then $g_A(z)$ is (up to normalization by $z^n$) the permanent of a biased version
of $A$ with mean $1/z$. 
Given this polynomial we ask - what is the distribution of roots of $g_A(z)$?
Our goal is to show that we can find a sequence of overlapping disks, whose radii is not too small, that allows us to
apply {\it analytic continuation} (see e.g. \cite{Ahlfors}) from the value of the function $g_A(z)$ at $z=0$ to $g_A(z)$ for some $|z|\gg1$.

A useful technique in the analysis of random polynomials is Jensen's formula which states that for an analytic function $g(z)$ with zeros $z_1,\hdots, z_n$
that is analytic in the complex disk of radius $r$ centered at the origin and $g(0)\neq 0$ we have
\be
\int_0^{2\pi} \ln(|g_A(r e^{i\theta})|) \frac{d\theta}{2 \pi}  - \ln(|g(0)|)  = \sum_{|z_j| \leq r} \ln \frac r {|z_j|}
\ee
In order to apply Jensen's formula we observe that the left hand side of the formula, for zero mean random matrices $A$, and $g_A(z) = \Per(J + z A)$
is essentially bounded from above by the the logarithm of the second moment $g_A(r)$.
We prove in Lemma \ref{lem:moment2} (which provides a simplified version of  a proof by Rempa{\l}a and Weso{\l}owski \cite{RW04})
that this second moment 
is upper-bounded by a term that scales at most exponentially with the 
square of $r$:
\be
\mathbf{E}_{A}[|g_A(r)|^2] \leq (n!)^2 \cdot e^{r^2}
\ee
Together with Jensen's formula this bound implies two interesting facts about the roots of $g_A(z) = \Per(J + z A)$ summarized in Proposition \ref{prop:root}:
that typical matrices $A$ are such that $g_A(z)$ has no roots inside the disk of radius $|z|$ for $|z|\ll1$, and very few (say $O(\ln(n))$) roots
in the disk of radius $|z| = \sqrt{\ln(n)}$.

These two facts imply together the existence of analytic curves arching from $z=0$ to some value $z$, $|z|\gg1$ 
with the following property: For a typical $A$ from the distribution, these curves are at distance at least some small $\eps>0$ from any root of $g_A(z)$.
These curves are depicted in Figure \ref{fig:tubes}.
Proposition \ref{prop:root} implies that most curves of this form avoid all roots of $g_A(z)$ with margin at least $\eps$ for most $A$'s, so
our algorithm samples such a curve at random and use it to interpolate the value of $g(z)$ from $z=0$ to a large value of $|z|$.

\subsection{Turning into an algorithm}

To recap the previous section, our overall strategy is to compute a Taylor series expansion of the logarithm of a polynomial $g_A(z)$ related to the permanent of a random matrix $A$.
In order to do that, we characterize the location of the roots of $g_A(z)$ for a typical matrix $A$ which allows us to find simple curves in the complex plane which
are at distance at least some small but large enough $\eps>0$ from any root of $g_A(z)$ for most $A$'s.
This would imply that for such matrices $A$ the function $f_A(z) = \ln(g_A(z))$ is analytic on any point along these curves, up to radius $\eps$.

However, it is not immediately clear that this analytic continuation strategy can be turned into an algorithm. Suppose $g_A(z)$ is root-free within $\eps$-neighborhood of each point of the segment $[0,1]$.
In his work \cite{Bar16} Barvinok suggested composing $g_A(z)$ with an auxiliary polynomial $\phi$ corresponding to $e^{O(1/\eps)}$ terms in the Taylor expansion of $\eps \ln \frac{1}{1-z}$ around $z=0$. He showed that $g_A \circ \phi$ has indeed no roots inside a disk of radius $1 - e^{-O(1/\eps)}$. See lemma 8.1 of \cite{Bar16} and Section 2.2 of \cite{B16} for more details.

%

It is somewhat less clear however, whether one can use the auxiliary map strategy for tubes around more elaborate curves
like a quadratic curve, or a piecewise linear curve (which we consider in this work), and whether such maps would result in good error parameters. In order to extend Barvinok's approach to such curves one needs to compose $\phi$ with another auxiliary low-degree polynomial map which maps the $\eps$ neighborhood of $[0,1]$ to an $O(\eps)$ neighborhood of these curves.
We however use a different method which allows us to interpolation along an arbitrary (even non-differentiable) curve. 
We name this complementary algorithmic technique Computational Analytic Continuation.

\subsubsection{Computational analytic continuation (CAC)}

CAC is an algorithm that for a degree $n$ polynomial $g(z)$ computes the value of $f(z) = \ln(g(z))$ at some value $z$ given an oracle access to derivatives of $f(z)$ at $z=0$,
and a curve $\gamma$, $\gamma(0)=0, \gamma(1)=z$ that is at distance at least some $\eps>0$ from all roots of $g(z)$.

Let $g(z)$ be a univariate polynomial $g:\C \rightarrow \C$ and assume we have an oracle which computes the derivatives of this function at $z =0$. 
Let $\gamma$ be some curve in the complex plane that is at distance at least $R$ from any root of $g(z)$.
We would like to approximate $f(\sigma) = \ln(g(\sigma))$ for some complex number
$\sigma = \gamma(1) \neq 0$,
using a Taylor series expansion of order $m$, that is small as possible.
For simplicity, we assume that $\gamma$ is piece-wise linear and
divide each segment of $\gamma$ into small intervals. We denote the entire sequence of intervals as $\Delta_1,\hdots, \Delta_t$.
For each $i$ the length of $\Delta_i$ is at most $R/\beta$ for all $i$, for some $\beta>1$.

We then use a sequence of iterations to compute many derivatives at each step:
at the first step $z=0$ we compute some $m$ derivatives, where $m$ is suitably chosen for small approximation error.
Then at the next step, we compute $ O(m/\ln(m))$ derivatives at point $\Delta_1$.
The update rule of $k$-th derivative at step $i$, denoted by $\hat{f}^{(k)}_i$  is merely the Taylor series approximation
of the $k$-th derivative $f^{(k)}$ as an analytical function to order $m$:
\be
\hat{f}^{(k)}_i \leftarrow \sum_{j=0}^m \frac{f^{(k+j)}_{i-1}\Delta_1^k}{ k!}.
\label{eq:8}
\ee
In general at each step $i$ we compute $s_i$ derivatives where the number of derivatives
compute at each step is reduced sharply with $i$:
\be
s_i \approx O(s_{i-1} / \ln(s_{i-1})).
\ee
We choose $m$ sufficiently large so that $s_t$ is sufficiently large to imply a $1/\poly(n)$ additive error in the final Taylor series approximation of $f^{(0)}_t$.
 Intuitively, if $t$ is the number of overlapping disks, then we need to fix $m \approx \ln n \times (\ln \ln n)^t$.

Since $\gamma$ is at distance at least $R$ from any root of $g(z)$, and the step sizes $\Delta_i$ are chosen to be 
sufficiently small compared to the convergence radius $R$ around the previous point $\sum_{j=0}^{i-1} \Delta_j$
it follows that the Taylor series  converges quickly to the true function value at each step.
We prove a quantitative estimate:

\begin{lemma}(Sketch)
Let $\Delta_{min} = \min_i |\Delta_i|$ and suppose that
$\Delta_{min} \leq R/\beta$ where $R$ is the convergence radius of $\ln(g(z))$ around point $\sum_{j=0}^{i-1} \Delta_j$,
minimized over all $1 \leq i < t$.
Consider the update rule in equation \ref{eq:8} that computes $s_i$ derivatives $f^{(k)}_i$ using $s_{i-1}$ previously
computed derivatives $f^{(k)}_{i-1}$.
Suppose that $s_0 = \ln \frac{n \sigma}{\delta \Delta} O(\ln 1/\Delta_{min})^{t}$ for some error parameter $\delta = 1/\poly(n)$.
Then 
\be
\left| \hat{f}^{(0)}_t - f(\sigma) \right| \leq \delta.
\ee
\end{lemma}
 
We then show that for a specific choice of parameters $\sigma = O(\ln n)$, $\Delta_{min} > 1/\poly\log (n)$, and $t= \frac{\ln\ln n}{\ln \ln \ln n}$ we get an inverse polynomial error by using a poly-logarithmic number of derivatives.
Since $e^{f^{(t)}} = g(\sigma) = \Per(J + \sigma A)$ for zero mean random matrix $A$
then $\sigma^{-n} e^{\hat{f}^{(0)}_t}$ is a $1 + 1/\poly(n)$ multiplicative approximation of the same random matrix but with
vanishing mean, i.e. $\mu = 1/\sigma$.

\begin{figure}
\center{
 \epsfxsize=3in
 \epsfbox{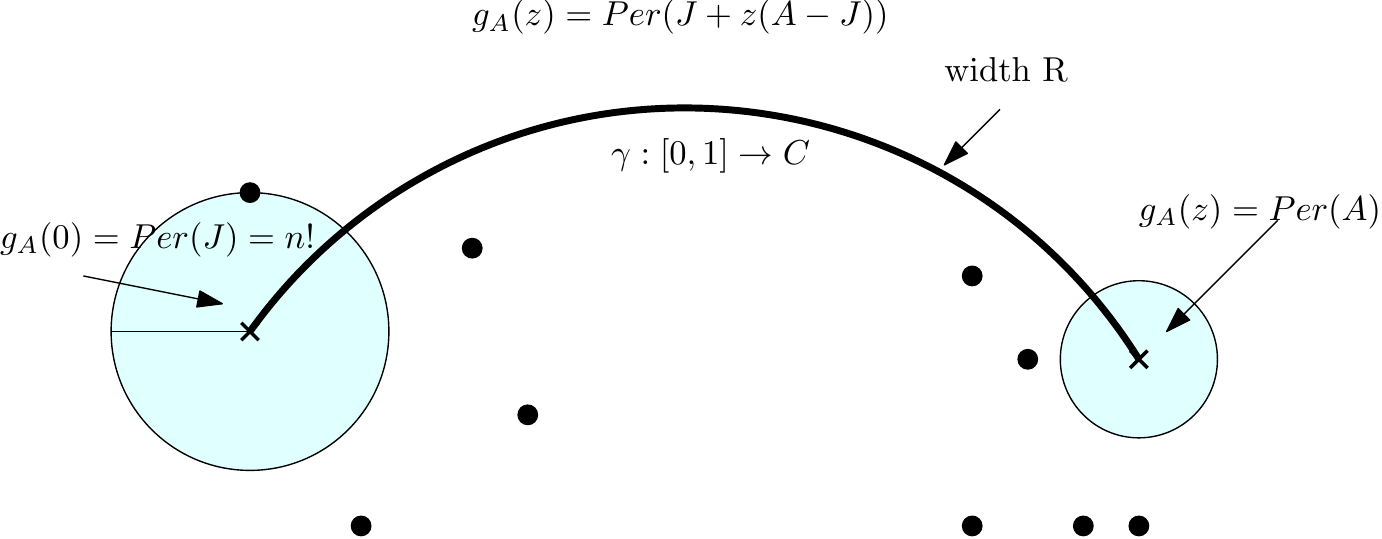}
 \caption{\footnotesize{The curve $\gamma$ connects $z=0$ and some value $z \neq 0$ and is at distance at least $R$
 from any root of $g_A(z)$.
 Note that interpolation along the curve in small segments can reach a lot further than taking a single interpolation
 step which is bounded by blue-shaded region - dictated by the location of the nearest root to the point $z=0$. 
 }\label{fig:zoo}}}
\end{figure}

\subsection{Discussion and future work}\label{Sec:further}

Our study extends the line of work pioneered by Barvinok that views the computation of the permanent of a matrix $A$
as a question about a related complex-valued polynomial $g_A(z)$.
In this work we allow $A$ to be a random matrix and hence recast the question about the permanent of a random matrix $A$
as a question about the location of roots of a random polynomial related to $A$.
We characterize the behavior of these polynomials for some random matrices $A$, and then provide
an algorithmic technique which allows us to turn this knowledge into a quasi-polynomial algorithm
for approximating the permanent of these matrices.

For a while now, it has been a folklore notion that the permanent is difficult to compute for general matrices
mainly because of the "sign problem" - namely the fact that entries with opposing signs generate intricate interference
patterns.
Such matrices avoid by definition the regime of matrices for which efficient algorithms are known, most prominent 
of which is the algorithm of Jerrum, Sinclair and Vigoda \cite{JSV04} for non-negative entry matrices.
Our work, we believe, places this notion in serious doubt - we show natural random ensembles with very little (i.e. vanishing) correlation
between any pair of entries - and yet, we are able to approximate the permanent for such matrices quite efficiently.
Hence, it seems that if in fact approximation of the permanent on average is a difficult task, it must be due to another, yet uncharacterized
phenomenon.
Furthermore, our study makes the hardness of approximating the permanent on average to an even more intriguing problem:
it seems that several natural ensembles do in fact admit an efficient approximation - but is it the case for other ensembles?
Most notably, one would like to consider the case of zero mean complex Gaussian matrices, the presumed hardness of which is the
basis for the ${\rm BosonSampling}$ paradigm, discussed in the following section.

\subsubsection{Implications to ${\rm BosonSampling}$}\label{sec:BS}

In \cite{AA13} the authors consider the following  computational problem:
\begin{definition}[{\rm GPE}$_{\times}^\mu$] Given $A \sim \mathcal{N}^{n\times n}(\mu,1,\C)$, $\epsilon, \delta$, output a number $Q$ such that with probability at least $1-\delta$, $\Big|Q- \Per(A)\Big|\leq \epsilon |\Per(A)|$ in $\poly(n,1/\epsilon,1/\delta)$.
\label{GPE1}
\end{definition}

\noindent
They conjecture that
\begin{conjecture}\label{conj:2}
${\rm GPE}_\times^0$ is $\#P$-hard to compute.
\end{conjecture}

Together with another conjecture on the anti-concentration of the permanent of complex Gaussian matrices
this conjecture implies that
\noindent BPP simulation of the linear-optical experiment called ${\rm BosonSampling}$ 
to within total variation distance
implies collapse of the polynomial hierarchy to the third level,
thereby establishing a so-called ``quantum supremacy'' of the outcomes of these physical experiments.
Using the same anti-concentration assumption on the permanent of zero mean Gaussian matrices
we explain 
 (see appendix \ref{app:bs}) that in fact the above conjecture is true also for complex Gaussian matrices with
mean $\mu = 1/\poly(n)$:
\be
\exists \mu = n^{-\Omega(1)}, \quad GPE_\times^{0} \preceq GPE_\times^\mu.
\ee
On the other hand, our main theorem implies that
\be
{\rm GPE}_\times^{1/\poly\log\log (n)} \in \mbox{DTIME}\left(2^{\poly\log(n)}\right),
\ee
and hence ${\rm GPE}_\times^{1/\poly\log\log(n)}$ is very unlikely to be $\# P$-hard.
This raises the following intriguing question: It seems that the hardness of the permanent
of complex Gaussian matrices (or general random matrices for that matter)
is not due to the common intuition that the different signs of the entries
prohibits combinatorial treatment of the matrix, as a graph, in the spirit of \cite{JSV04}. 
Hence, if indeed ${\rm GPE}_\times^0$ is hard there must exist another phenomenon governing the behavior of the permanent
of complex Gaussian matrices with mean values between  $\mu = 1/\poly(n)$ and $\mu = 1/\poly\log\log(n)$ which makes computation intractable.


\subsubsection{Reducing the mean value}
A natural next step for our approach is to attempt to further increase the value of $z$ for which we evaluate $g(z)$.
Approximating $g_A(z)$ for typical $A$ at $|z| = 1/\mu$ implies an approximation of the permanent for a random matrix with mean $\mu$ and variance $1$.

However, one can see from the upper-bound on the interpolation error above that 
in order to achieve sufficiently small error the number of derivatives we initially compute must scale doubly
exponentially in the ratio $\sigma/\Delta$, namely the ratio of the interpolation length $\sigma$, and the step size $\Delta$.
Since $\Delta \sim 1/N_\sigma$ where $N_\sigma = \Omega(\sigma^2)$ is the number of roots in the disk of radius $\sigma$, then 
$\Delta \sim 1/\sigma^2$ which implies that the number of required derivatives is exponential in $\poly(\sigma)$.

Thus to improve on our technique new ideas would be required which would make a more economic
use of the derivatives computed, and not ``discard'' at each step $i$ a fraction $1/\ln(s_i)$ of the $s_i$
derivatives computed at that step.
Another approach would be to tighten the bound on the number of roots inside the disk of a given radius
for the polynomial $g_A(z)$ or some other related polynomial.

\subsubsection{Anti-concentration of the permanent}

We point out that our algorithm, and also that of \cite{Bar16,BS17}, are not merely a sequence of computational steps,
but actually show that the permanent of typical complex Gaussian (or more generally, random) matrices $A\sim {\cal N}(\mu,1,\C)^{n\times n}$
are well approximated
by a low-degree polynomial in the entries of $A\sim {\cal N}(0,1,\C)^{n\times n}$ - in fact a polynomial of degree $\poly\log(n)$.
Such a statement strongly indicates that for this range the permanent of complex Gaussian matrices
is anti-concentrated, since it is well-known by a theorem of Carbery-Wright \cite{CW01} that any polynomial
of low-degree in standard i.i.d. complex Gaussian variables is anti-concentrated.

However, we have been unable to use this theorem directly to derive a formal anti-concentration statement.
We note that the anti-concentration of the Gaussian permanent is also a necessary conjecture
of the ${\rm BosonSampling}$ paradigm in \cite{AA13, AC} along with the conjectured $\# P$-hardness of this problem.
Hence, intriguingly it seems that the conjectures on the computability, and statistics (i.e. anti-concentration)
of the permanent of the complex Gaussian matrices are closely related: 
Our results suggest that
for mean values $1/\poly\log(n)$ the anti-concentration conjecture holds, 
but ironically this is because the second conjecture does not hold for this regime  - i.e.
the permanent is relatively easy to compute via a polynomial of low degree.

\section{Preliminaries}

\subsection{Notation}
$\C$ denotes the complex numbers, and $\R$ denotes the real numbers.
${\cal N}(\mu,\sigma,\C)$ is the complex Gaussian distribution
of mean $\mu$ and variance $\sigma^2$. 
${\rm Bern}(\mu)$ is the biased-Bernoulli random variable - it is $1$ w.p. $1/2$ and $- 1 + \mu$ w.p. $1/2$.
$\Per(X)$ is the permanent of matrix $X$.
$\ln(x)$ is the natural complex logarithm of $x$, 
defined up to an additive factor $2\pi i k$ for integer $k$.
${\cal B}_r(z)$ denotes the closed disk in the complex plane of radius $r$ around point $z$.
For computational problems $A,B$ we denote $A \preceq B$ if there exists a poly-time reduction from $A$ to $B$ -
i.e. $A$ is no harder than $B$.

\begin{definition}

\textbf{Random Matrix}

\noindent
An $n\times n$ matrix $A$ 
is called a random matrix and denoted by
$A\sim \mathcal {D} (\mu, \sigma^2)^{n\times n}$,
if each entry of $A$ is independently drawn from some complex valued distribution with mean $\mu$ and variance $\sigma^2$.
\end{definition}
The entries of $A$ in the above definition do not have to be identically distributed. 
We denote the distribution of complex Gaussian matrices with mean $\mu$ and variance $\sigma^2$ with $\mathcal{N} (\mu, \sigma^2, \C)^{n\times n}$.

\section{Permanent-interpolating-polynomial as a random polynomial}

As described in Section \ref{sec:barvinok} recent studies designed algorithms
for evaluating multi-variate functions like the permanent or Ising model by
considering a related univariate polynomial $g(z)$.
These schemes used this polynomial to interpolate the value
of $\ln(g(z))$ at some point of interest $z = z_0$ using knowledge of the derivatives
at $z=0$.

For example, in his work, Barvinok \cite{Bar16} used the polynomial $g(z) = \Per(J \cdot (1-z) + A \cdot z)$.
In a more recent work \cite{LSS17} the authors choose a different polynomial $g(z) = Z_{\beta}(z)$ - namely the Ising partition function.
In both of these works the authors characterized the location of the roots of $g(z)$ in order to establish
that $\ln(g(z))$ is analytical in the region of interest.
Indeed, in his work, Barvinok showed that $g(z)$ has no roots in the unit disk for any matrix $A$ that satisfies $\max_{ij} |A_{ij} - J| \leq 0.5$
Likewise, the polynomial in the work of \cite{LSS17} the polynomial $g(z)$ was shown to have no roots inside the unit disk using
the Lee-Yang theorem. 

In our work we consider the polynomial
\be
g_A(z):= \Per(J + z \cdot A),
\ee
and then analyze it as a random polynomial in order to gain insight into the distribution of its roots. 
The choice of this polynomial has a more natural interpretation in the context of random matrices:
for a zero mean random matrix $A$, the value of $g_A(z)/z^n$ for nonzero $z$ is
the value of the permanent of a matrix drawn from the ensemble $A$ when shifted with a mean of $1/z$ and variance $1$. Another reason why we choose this polynomial is that it is easier to bound its roots compared to $\Per((1-z) J + z \cdot A)$.

We begin with the following definition
\begin{definition}\label{def:avg}[Average sensitivity]
\noindent
Let $g_A(z)$ be a random polynomial where $A$ is a random matrix.
For any real $r>0$ the stability of $g_A(z)$ at point $r$ is defined as
\be
\kappa(r) \equiv \kappa_g (r):= \E_\theta \E_A \left[ \frac{|g_A(re^{i\theta})|^2}{|g_A(0)|^2} \right],
\ee
where $\E_\theta [\cdot] = \int_{\theta =0}^{2\pi} [\cdot]\frac{d \theta}{2\pi}$ is the expectation over $\theta$  from  a uniform distribution over $ [0,2\pi)$.
\end{definition}

We begin with an upper-bound on the average sensitivity of the permanent of a random matrix
that can also be derived from the work of Rempa{\l}a and Weso{\l}owski \cite{RW04} (see Proposition 1). 


\begin{lemma}\label{lem:moment2}

\noindent
Let $A \sim \mathcal{D}^{n\times n} (0,1)$.
Then
\be
\kappa_g (r) \leq e^{r^2}.
\ee
\label{lemma:moment}
\end{lemma}
A somewhat simpler and more intuitive proof of this lemma is given in \cite{EM17} (see Lemma 12).

Our next result is to relate the number of roots of $g_A$ inside a disk of certain radius around origin to its average sensitivity around the boundary of that disk. Jensen's formula provides this connection.
\begin{proposition}\label{prop:root}
Let
\be
g_A(z) = \Per(J + z \cdot A),
\ee
and $A \sim \mathcal {D}^{n\times n} (0,1)$. Then if $N_r$ is the number of roots of $g_A$ inside a disk of radius $r$, $\E_A [N_r] \leq 4 r^2$. In particular:

\begin{enumerate}
\item\label{it:large} 
Let $r > 1$. With probability at least $1- \frac{1}{r}$ the polynomial $g_A(z)$ has at most $4 r^3$ roots
inside the disk with radius $r$ around $z=0$. 

\item\label{it:small} 
Let $r < 1/2$. With probability at least $1-  4 r^2$ there are no roots inside the disk with radius $r$ around $z=0$.

\end{enumerate}
\label{proposition:roots}
\end{proposition}

The above claim immediately implies:
\begin{corollary}
Let $\eps$ be real number $0 < \eps< 0.5$. For at least $1- 3 \eps$ fraction of matrices $A \sim \mathcal {D}^{n\times n} (0,1)$ $g_A$ has no roots inside $\mathcal {B}_\eps (0)$ and has at most $32/\eps^3$ roots inside $\mathcal{B}_{2/\eps}(0)$.

 \label{cor: safedisks} 
\end{corollary}
\begin{proof} The proof follows by a union bound on the two items in Proposition \ref{proposition:roots}. In particular, the probability that $\mathcal {B}_\eps (0)$ is not root free is at most $4 \eps^2$, and the probability that $\mathcal{B}_{2/\eps}(0)$ has more than $32/\eps^3$ roots is at most $\eps/2$, so using union bound for $\eps < 0.5$ the probability of error amounts to $\eps/2 + 4 \eps^2 < 3 \eps$.
\end{proof}

\begin{proof} (of Proposition \ref{proposition:roots})
We use Jensen's formula. Let $g: \C \rightarrow \C$ be an arbitrary polynomial
 such that $g_A(0) \neq 0$. Let $N_r$ be the number of zeros of $g$ inside a disk of radius $r$, and let $z_1, \ldots, z_{N_r} \in \C$ be these zeros. Then Jensen's formula is
\be
\sum_{|z_j| \leq r} \ln  \frac{r}{|z_j|} + \ln |g_A(0)| = \E_\theta \ln |g_A(r e^{i\theta})|,
\ee
where we have used the notation
\be
\int_{\theta=0}^{2\pi} [\cdot] \frac{d\theta}{2\pi} =: \E_\theta [ \cdot ].
\ee
Let $0 < \delta < 1$ be a real number. We first use the following bound
\bea
\sum_{|z_j| \leq r} \ln  \frac{r}{|z_j|} &\geq& \sum_{|z_j| \leq r (1-\delta)} \ln  \frac{r}{|z_j|}\\
&\geq& \sum_{|z_j| \leq r (1-\delta)} \ln  \frac{r}{r(1-\delta)}\\
&\geq& \delta \cdot N_{r (1-\delta)}.
\eea
We now pick $A \sim {\cal D} (0,1,\C)^{n\times n}$ and view the variables in the Jensen's formula above as random variables depending on $A$.
 By Jensen's formula
\ba\label{eq:Nr}
\delta\cdot N_{r (1-\delta)}
&\leq
\E_\theta \left[\ln(|g_A(r e^{i\theta})|)   \right] -  \ln(|g_A(0)|),\\
&= \E_\theta \left[\ln \Big(\Big|\frac{g_A(r e^{i\theta})}{n!}\Big| \Big)   \right],\\
&=\frac 1 2 \E_\theta \left[\ln \Big(\Big|\frac{g_A(r e^{i\theta})}{n!}\Big|^2 \Big)   \right],\\
&= \frac 1 2 \cdot \E_\theta [\ln b(r,\theta)],
\label{eq:44}
\ea
where 
$b(r,\theta):= \Big|\frac{g_A(r e^{i\theta})}{n!}\Big|^2$. 
Then by Lemma \ref{lem:moment2}
\be
\forall \theta \quad
\E_{A} \left[ b(r,\theta) \right] \leq e^{r^2},
\ee
and so in particular
\be
\E_{\theta,A} \left[ b(r,\theta) \right] \leq e^{r^2}.
\ee
So by concavity of the logarithm function
\ba
\delta\cdot \E_A N_{r (1-\delta)} &\leq \frac 1 2 \E_{\theta, A} \left[ \ln\Big(b(r,\theta)\Big)\right],\\
&\leq \frac 1 2  \ln \Big(\E_{\theta, A}\left[b(r,\theta)\right]\Big),\\
&\leq r^2/2.
\ea
As a result, choosing $\delta = 1/2$ and doing the change of variable $r \leftarrow r/2$,
\be
\E_A N_{r} \leq 4 r^2.
\label{eq:expected}
\ee
Thus by Markov's inequality, when $r>1$:
\be
\Pr_A (N_r \geq 4 r^3) \leq \frac 1  r.
\ee
Next we consider the case when $r < 1$. In this case we can directly use the following Markov's inequality using equation \ref{eq:expected}:
\be
\Pr(N_r \geq 1) \leq 4 r^2.
\ee
Note this bound is useful only when $r < 1/2$.
%
%
%
\end{proof}

\subsection{Root-avoiding curves}

We now use the above insight on the distribution of roots of the random polynomial $g_A(z)$
to edge closer to an algorithm.
We define:
\begin{definition}\label{def:free}

\textbf{Root-free area}

\noindent
A subset $S \subseteq \C$ is root-free w.r.t. polynomial $g(z)$
if it contains no roots of $g$.
\label{definition:rootfreedisk}
\end{definition}

\noindent
\begin{definition}

\textbf{Tube of width $w$ around a curve}

\noindent
Let $w > 0$ be a real number, and let $\gamma: [0,1] \to \C$ denote some parameterized curve in the complex plane.
The tube of width $w$ around $\gamma$
denoted by ${\cal T} (\gamma, w)$,
is the set of points defined as $\cup_{x \in \gamma} \mathcal{B}_w(x)$, where $\mathcal{B}_w (x)$ is the closed $w$-ball centered around $x\in \C$. In other words, for each point on the curve we include the $w$-ball around it in the tube.
\label{definition:fattening}
\end{definition}

We will denote by $L(a,b)$ the linear segment in $\C$ between $a,b\in \C$.

\begin{lemma}\label{lem:root2}

\textbf{Finding a root avoiding tube}

\noindent
Let $0 < \eps < 0.1$ and $A\sim {\cal D}(0,1)^{n\times n}$.
Fix $w = \pi\eps^6$.
For numbers $a\in \C,b\in \R$ consider the following curve
\be
\gamma_{a,b}(t) =
  \begin{cases}
    at  & t \in [0,\frac 1 2) \\
    a (1 - t) + 2 b (t-\frac 1 2) &  t \in [\frac 1 2, 1]
  \end{cases}
\ee
There exists $a\in \C$, $|a| \leq 2 \eps$ 
and $b\in \R$, $b\in [1/\eps,1/\eps + 2 \eps]$ 
such that the tube
${\cal T}(\gamma_{a,b},w)$ is root-free w.r.t. $g_A(z)$
for a fraction
at least $1 - 4 \eps$ over choices
of $A$. The total length of such a tube is at most $2/\eps$.
\end{lemma}

%

\begin{proof} 
Fix $M = 32/\eps^5$. 
For each $j\in [M]$ define the piece-wise linear curve $\gamma_j$ 
made of 2 segments
\bea
L(0, 2\eps + i \cdot 2\eps \tan(2\pi j / M)),
\eea
and,
\bea
L( 2\eps + i \cdot 2\eps \tan(2\pi j /M), 1/\eps\cdot \tan(2\pi j / M)).&&
\eea
See figure \ref{fig:tubes}.
\begin{figure}
\center{
 \epsfxsize=3in
 \epsfbox{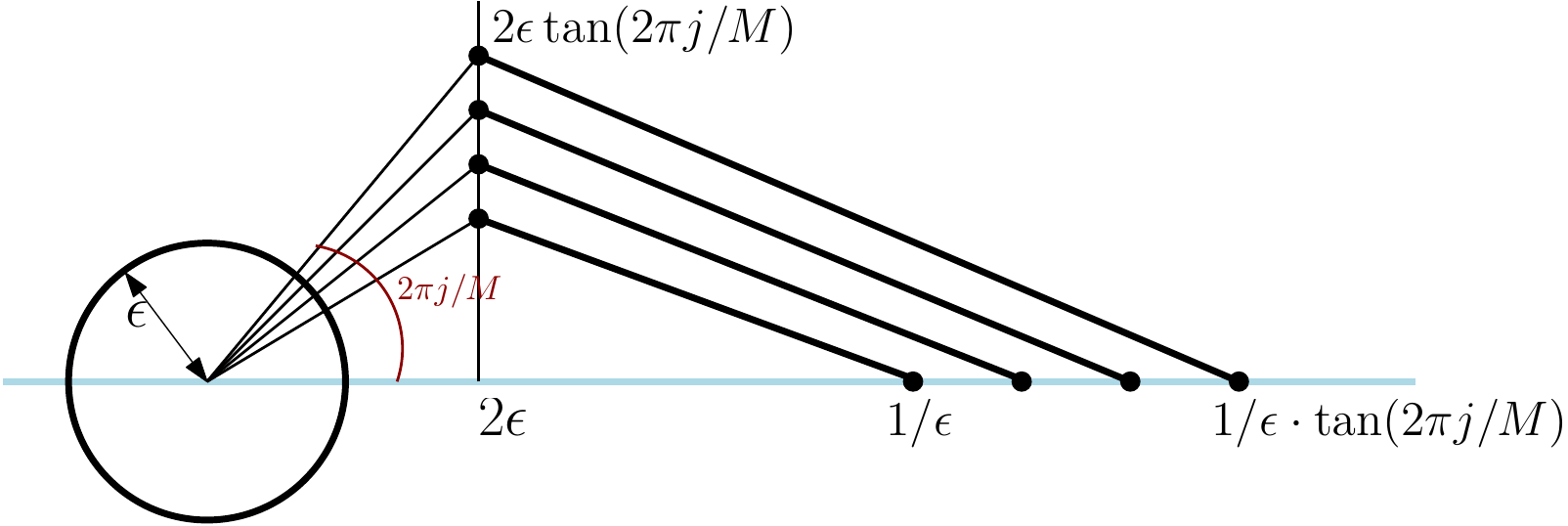}
 \caption{\footnotesize{
 The family of 2-piecewise linear curves interpolating from $z=0$ to some point $z$ on the real line.
 The curves branch out from the unit circle at an angle between $\pi/4$ and $\pi/4 + \eps$.
 Each curve starts out at an angle of $2\pi j/M$ for some integer $j$ and hits the imaginary axis at $\Re (z) = 2\eps$.
 Hence the imaginary magnitude of the end-point of the first segment in the $j$-th curve is $2 \eps \tan(2\pi j/M)$.
 After hitting the imaginary axis at $\Re{(z)} = 2\eps$, they descend back to the real line in {\it parallel},
 thus hitting the real line at different points.
 The bottom-most curve hits the real line at $z = 1/\eps$.
 Note that by this definition tubes of weight $w$ around each curve do not overlap outside the ball of radius $\eps$
 and in particular when they hit the imaginary axis at $\Re(z) = 2 \eps$.
}\label{fig:tubes}}}
\end{figure}
Note that specifically for a small subset of indices $j\in [M/8,\hdots, M/8 +  \eps M]$ the value of $b$ is locked in a tight range:
\be\label{eq:ab}
|a| \leq 2 \eps, \quad b\in [1/\eps, 1/\eps + 2 \eps]
\ee

Let ${\cal T}_j  = {\cal T}(\gamma_j,w)$.
Then by Equation \ref{eq:ab} the union of these tubes is contained inside a ball that is not too large:
\be\label{eq:disjoint}
\bigcup_j {\cal T}_j \subset \mathcal{B}_{2/\eps} (0).
\ee
and that by our definition of $w$ these tubes are disjoint outside a ball of radius $\eps$:
\be\label{eq:disjoint2}
\forall j\neq k \quad {\cal T}_j \cap {\cal T}_k - {\cal B}_{\eps}(0) = \emptyset
\ee
Let $E$ the event of no roots at distance at most $\eps$ to the origin and few roots inside the disk of radius $2/\eps$:
\be
E = \{ A:  
N_{\eps}=0 \quad \wedge  \quad N_{2/ \eps} \leq 32/\eps^3\}
\ee
By Corollary \ref{cor: safedisks}
\be\label{eq:E}
\Pr(E) \geq 1 - 3 \eps.
\ee
Condition on $E$ and for each $j\in [M/8,\hdots, M/8 + \eps M]$ and matrix $A$ let $m_{j,A}$
denote the number of roots of $g_A(z)$ inside ${\cal T}_j - \{0\}$.
Then by Equation \ref{eq:disjoint2} and the definition of $E$ we have
\be
\forall A \in E \quad
\sum_{j = M/8}^ {M/8 + \eps M} m_{j,A} \leq N_{2/\eps}
\ee
and so for uniform random $j$ the average number of roots is small
\be
\forall A \in E \quad
\mathbf{E}_j \left[ m_{j,A} \right] \leq N_{2/\eps}/(\eps M) \leq \eps
\ee
and so in particular this holds for average matrix $A$ (conditioned on $E$):
\be
\mathbf{E}_{A|E} \mathbf{E}_{j}  \left[ m_{j,A} \right] \leq \eps
\ee
so by linearity of expectation
\be
\mathbf{E}_j \mathbf{E}_{A|E}   \left[ m_{j,A} \right] \leq \eps
\ee
which implies that there exists $j_0\in \{M/8,\ldots,M/8 + M\eps\}$ (an index that minimizes $\mathbf{E}_{A|E}   \left[ m_{j,A} \right]$) such that
\be
\mathbf{E}_{A|E}   \left[ m_{j_0,A} \right] \leq  \eps
\ee
hence for $j_0$ we have
\be
\Pr_{A|E}( m_{j_0,A} >0 ) \leq  \eps
\ee
and so by the union bound with the probability of $E$ from Equation \ref{eq:E}
\be
\Pr_{A}( m_{j_0,A} >0 ) \leq \eps + 3 \eps = 4 \eps.
\ee
Note that the total length of this tube is at most $2/\eps$.

\end{proof}

%

\section{Computational analytic continuation}


In this section we pose the question of analytic continuation as a computational task and devise a method
to derive the value of the function $g(z)$ at some point $z = r$ using it's derivatives at $z=0$,
assuming that $g$ has a root-free tube around a curve $\gamma$, $\gamma(0)=0, \gamma(1) = z$.
We require the following result of Barvinok:
\begin{lemma}\label{lem:bar}\cite{Bar16}

\noindent
\textbf{Efficient computation of derivatives}

\noindent
\begin{enumerate}
\item 
Let $A$ be an $n \times n$ complex matrix and let $g_A(z):= \Per(J + z A)$, where $J$ is all ones matrix. 
There exists a deterministic algorithm that runs in time $n^{O(l)}$ and computes the $l$-th derivative of $g_A(z)$ at point $z = 0$.
\item
Let $g(z)$ be a polynomial of degree $n$.
Given an algorithm 
that runs in time $t$ and computes the first $\ell$ derivatives of $g(z)$ at point $z=0$,
one can compute in time $O(\ell^2 t)$ the first $\ell$ derivatives of $f(z) = \ln(g(z))$ at $z=0$.
\end{enumerate}

\end{lemma}

\noindent
We also need the following technical numerical result the proof of which is deferred to the appendix \ref{section:lem20}.
\begin{lemma}\label{lem:int1}
For all $m \geq 2 l$ and $\beta \geq e$
\be
\sum_{k=m}^{\infty} \beta^{-k} \cdot k^l
\leq 3\cdot \beta^{-m} m^l.
\ee
\end{lemma}

We now present a deterministic algorithm for computing the analytic continuation (see \cite{Ahlfors}) of a degree $n$ polynomial 
$g: \C \rightarrow \C$.

\begin{mdframed}
\begin{algorithm}[Computational analytic continuation]
\label{algorithm:main}
\begin{enumerate}

\noindent
\item
\textbf{Input:} 
An oracle $\mathcal{O}_g$ that takes a number $m$ 
as input and outputs the first $m$ derivatives of $g$ at $z = 0$.
$t$ complex numbers $\Delta_1,\hdots,\Delta_t$, 
a number $\beta>1$,
precision parameter $\delta>0$, and integer $m$ - the number of derivatives
computed at the $0$-th step.
\item
\textbf{Fixed parameters:}

\begin{enumerate}
\item
$\Delta_{min} = \min_i |\Delta_i|$.
\item 
$y_0=0$ and $y_i= \sum_{j=1}^{i-1} \Delta_j $ for each $1 \leq i \leq t-1$ \hfill 
\end{enumerate}

%
%

\item
\textbf{Variables:}

\begin{enumerate}
\item $\hat{f}^{(l)}_i$ for $1 \leq l\leq m$ and $0 \leq i \leq t-1$ \hfill \% the $l$-th derivative of $f$ at $y_i$.
\item $s_i$ for $0\leq i \leq t-1$ \hfill \% the number of derivatives at each point $y_i$.
\item $s_0 \leftarrow m$.
\end{enumerate}

\item
\textbf{Main:}
\begin{enumerate}
\item\label{it:it1} Query $\mathcal{O}_g(m)$ to obtain $g^{(0)}(0),\ldots, g^{(m)}(0)$ 

\item\label{it:it2} Using Lemma \ref{lem:bar} and derivatives from step I 
compute 
$\hat{f}^{(l)}_0 \leftarrow f^{(l)}(y_0)$ 
for $1 \leq l\leq m$. 

\item For each $i = 0 \twodots t-1$:

\begin{itemize}
\item Set: $s_{i+1} \leftarrow \frac{\ln \beta}{2 } \frac{s_i}{\ln (2s_i/\Delta_{min})}.$

\item Compute:
$\forall k\in [s_{i+1}], \ \ \ 
\hat{f}^{(k)}_{i+1}  = \sum_{p = 0}^{s_{i} - k} \frac{\hat{f}^{(p + k)}_i}{p!} \Delta_i^{p}.$
\end{itemize}
\label{it:dyn}
\end{enumerate}

\item
\textbf{Output}:

Let $\hat{f}:= \hat{f}^{(0)}_t$ and return ${\cal O} = e^{\hat{f}}$.

\end{enumerate}
\vspace{5mm}
\end{algorithm}
\end{mdframed}

%

\small
\noindent

Prior to establishing the correctness of the algorithm, we define shifted versions of $g(z)$ as follows:
\be
\forall i\in [t] \quad
\tilde{g}_i(z) = g(z + y_i),
\ee
and
\be
\tilde{f}_i(z) = \ln(\tilde{g}_i(z)),
\ee
and denote $f_i^{(l)} = \tilde{f}_i^{(l)}(0)$. Note $y_i$'s are defined in algorithm \ref{algorithm:main}.
We need the following elementary fact which we leave without proof. 
\begin{lemma}
If the closest root of $g$ to the point $y_i$ in the complex plane is $\lambda$, then the closest root of $\tilde{g}_i$ to $z=0$ is also $\lambda$.
\label{lemma:closest}
\end{lemma}

\noindent
We now establish correctness:
\begin{claim}\label{cl:alg}

\textbf
{Correctness of algorithm \ref{algorithm:main}}

\noindent
Let $g(z)$ be a polynomial of degree at most $n$, and $f(z) = \ln(g(z))$.
Suppose the inputs to algorithm \ref{algorithm:main} satisfy the following conditions:
\begin{enumerate}
\item
Precision parameter: $\delta \geq n^{-c_1}$ for some constant $c_1>0$.
\item
Interpolation length: $\sigma_t:= \sum_i |\Delta_i| \leq c_2\ln(n)$ for some constant $c_2>0$.
\item
Minimal segment length: $\Delta_{min}= \min_i |\Delta_i| \geq \ln^{-c_3}(n)$ for some constant $c_3>0$.
\item
Number of iterations: $t  \leq c_4 \ln(n) \ln (n)/\ln \ln \ln(n)$ for some constant $c_4>0$.
\item\label{it:ratio}
For each $i$ the ratio of the distance of the closest root of $g(z)$ to $y_i$ to step size $|\Delta_{i+1}|$ is at least $\beta \geq e$.
\item
Number of derivatives at step $0$: $m \leq \ln^{c_5}(n)$ for some constant $c_5>0$.
\end{enumerate}
Then the following holds:
there exists a constant $c_0$ such that if the number of derivatives $m$ that the algorithm queries from ${\cal O}_g$ 
at step $0$ is at least
\be
\ln(n) \cdot (c_0 \cdot \ln\ln(n))^t,
\ee
then output of the algorithm satisfies
\be
{\cal O} = e^{\hat{f}} = g(y_t) \cdot (1 + {\cal E}), \ \ |{\cal E}| = O(\delta).
\ee
where ${\cal E}$ is a complex number.
\end{claim}

%
%

\begin{proof}
Let $f (z):= \ln (g(z))$. It is sufficient to show that 
\be\label{eq:delta}
\left|\hat{f}-f(y_t)\right| \leq \delta.
\ee


 Let $\hat{f}^{(k)}_i$ denote the approximation of the $k$-th derivative of $f$ at point $y_i$ obtained by the algorithm.
Using oracle $\mathcal{O}_g$ for $0 \leq l \leq s_0$ we can compute precisely the derivatives of $g$ at $y_0=0$ 
and hence using the first part of Lemma \ref{lem:bar} evaluate the derivatives of $f$ precisely at $y_0$:
\be
\hat{f}^{(l)}_0 \leftarrow f^{(l)} (y_0).
\ee
For $i = 0 \twodots t-1$ (in order) 
algorithm \ref{algorithm:main} computes the lowest $s_{i+1}$ derivatives using the first $s_i$ derivatives as follows:
\be\label{eq:alg1}
\forall l\in [s_{i+1}], \quad
\hat{f}^{(l)}_{i+1}  = \sum_{p = 0}^{s_{i} - l} \frac{\hat{f}^{(p + l)}_i}{p!} \Delta_i^{p}.
\ee
By assumption \ref{it:ratio} and Lemma \ref{lemma:closest} for each $1\leq i \leq t$ the function $\tilde{f}_{i-1}$ is analytical about point 
$y_{i-1}$ in radius $\beta |\Delta_i|$.
Hence, we can write the $l$-th derivative of $\tilde{f}_{i+1}(z)$ 
as the infinite Taylor series expansion of the $l$-th derivative of $\tilde{f}_i(z)$ evaluated at point $\Delta_{i}$:
\be\label{eq:ideal1}
f^{(l)}_{i+1} = \sum_{p=0}^{\infty} \frac{ f^{(p+l)}_i }{ p!} \Delta_i^p.
\ee
Let ${\cal E}^{(l)}_i$ denote the additive approximation error of the $l$-th derivative at step $i\in [t]$
 and $l\in [s_i]$.
\be
\mathcal{E}^{(l)}_i:= \Big| \hat{f}^{(l)}_i - {f}^{(l)} _{i}\Big|, \hspace{1cm}  \forall l\in [s_{i}]
\ee
also let $\delta_i$ denote the worst-case error for all derivatives at step $i$
\be
\forall i\in [t], \quad
\delta_i:= \max_{0 \leq l \leq s_i} \Big(\mathcal{E}^{(l)}_i\Big).
\ee
Using the triangle inequality on equations \ref{eq:alg1} w.r.t. \ref{eq:ideal1} we
$\forall i\in [t],  l\in [s_i]$ get:
\begin{align}
\nonumber
\mathcal{E}^{(l)}_i &\leq 
\sum_{p=0}^{s_{i-1} - l} \frac{|\hat{f}^{(p+l)} _{i-1}-{f}^{(p+l)} _{i-1}|}{p!} |\Delta_i|^{p}\\ 
&+ 
\sum_{p=s_{i-1} - l +1}^{\infty} \frac{|{f}^{(p+l)} _{i-1} |}{p!} |\Delta_i|^{p},\\ \nonumber
&= \sum_{p=0}^{s_{i-1} - l} \frac{\mathcal{E}^{(p+l)}_{i-1}}{p!} |\Delta_{i}|^{p}\\ 
&+ 
\sum_{p=s_{i-1} - l +1}^{\infty} \frac{|{f}^{(l+p)}_{i-1} |}{p!} |\Delta_{i}|^p,\\ \nonumber
&\leq \delta_{i-1} e^{|\Delta_i|}\\
&+ \sum_{p=s_{i-1} - l +1}^{\infty} \frac{|{f}^{(p+l)}_{i-1} |}{p!} |\Delta_i|^{p},\\
&=: \delta_{i-1} e^{|\Delta_i|} +\kappa_{i,l},
\label{equation:err}
\end{align}
where
\bea
\kappa_{i,l} &:=& \sum_{p=s_{i-1} - l +1}^{\infty} \frac{|{f}^{(p+l)}_{i-1} |}{p!} |\Delta_i|^{p},\\
&=& 
\sum_{p=s_{i-1} - l +1}^{\infty} \frac{|\tilde{f}_{i-1}^{(p+l)} (0) |}{p!} |\Delta_i|^{p}. 
\label{eq:kappa4}
\eea
At this point, we focus on placing an upper bound on $\kappa_{i,l}$. Fix any index $i$ and
let $z_1, \ldots, z_n$ be the roots of the shifted function $\tilde{g}_{i-1}$. Then
\be
\tilde{g}_{i-1} (z) = \tilde{g}_{i-1} (0) \left(1- \frac{z}{z_1}\right) \ldots \left(1- \frac{z}{z_n}\right). 
\ee
Therefore we can expand its logarithm as the infinite series:
\be
\forall k, \quad
\tilde{f}^{(k)}_{i-1} (0) = - \sum_{j=1}^n\frac{(k-1)!}{z^k_j}. 
\ee
Using these derivatives and the triangle inequality $\forall l\in [s_{i}]$ we can bound equation \ref{eq:kappa4}
\bea
\kappa_{i,l} &\leq& \sum_{j=1}^n \sum_{p=s_{i-1} - l +1}^{\infty} \frac{(l + p-1)!}{p!} \frac{|\Delta_i|^{p}}{|z|^{p + l}_j},\\
&\leq&  \sum_{j=1}^n \frac{1}{|z_j|^l}\sum_{p=s_{i-1} - l +1}^{\infty} (2p)^l \frac{|\Delta_i|^p}{|z|^{p}_j}.
\eea
The last inequality is by using $p \geq s_{i-1} - l \geq l - 1$, which is true because 
$l\leq s_i \leq s_{i-1}/2$ by the
choice of $s_i$ in the algorithm.
By assumption \ref{it:ratio} and Lemma \ref{lemma:closest} all roots of $\tilde{g}_{i-1}$ are located outside ${\cal B}_{|\Delta_i| \cdot \beta}(0)$:
\be
\forall j, \ \ |z_j| \geq \beta |\Delta_i|.
\ee
Therefore since $\beta \geq e$ in step \ref{it:ratio} of claim \ref{algorithm:main}
\bea
\kappa_{i,l} 
&\leq&  \frac{n}{|\beta \Delta_i|^l}\sum_{p=s_{i-1} - l +1}^{\infty} (2p)^l \beta^{-p},\\
& =& \frac{n}{(\beta |\Delta_i|/2)^l}\sum_{p=s_{i-1} - l +1}^{\infty} p^l \beta^{-p}.
\eea
One can check that our choice of $s_i$ in algorithm \ref{algorithm:main} satisfies $s_{i-1} - l + 1  >  2 s_i \ln \beta$, therefore we can use Lemma \ref{lem:int1}.
Using this lemma $\forall i\in [t], \ \ l\in [s_i]$ above equation implies:
\bea
\kappa_{i,l} &\leq& 3 \frac{n}{(\beta |\Delta_i|/2)^l} (s_{i-1}-l+1)^l \beta^{- s_{i-1} +l -1},\\ 
&\leq&  3 \frac{n}{( |\Delta_i|/2)^l} (s_{i-1})^l \beta^{- s_{i-1}},\\ 
&\leq&   3 \cdot n \left(\frac{2s_{i-1}}{|\Delta_i|}\right)^l \beta^{- s_{i-1}}, \\
&:=& \kappa_i.
\label{eq:kappa5}
\eea
Note the last bound does not depend on $l$ and we define as a number $\kappa_i$. 
Using the definition of $\kappa_i$ in equation \ref{eq:kappa5} equation \ref{equation:err} becomes
\be
\delta_i \leq \delta_{i-1} e^{|\Delta_i|} + \kappa_i.
\ee
we solve this recursion. Let $\sigma_i = \sum_{j=0}^i |\Delta_j|$. Note that since we compute the derivatives exactly at $y_0$ then $\delta_0=0$. Hence
\be
\delta_t \leq \sum_{i=1}^t \kappa_i e^{\sigma_t - \sigma_i}.
\ee
Our objective is to show that $\delta_t \leq \delta$. To do this it is enough to show that for all $i\in [t]$:
\be
\kappa_i e^{\sigma_t-\sigma_i} \leq \frac{\delta}{t}.
\label{eq:rec12}
\ee
Using definition \ref{eq:kappa5}, equation \ref{eq:rec12} is equivalent to
\be
e^{\sigma_t-\sigma_i} \frac{3\cdot n t}{\delta} \left(\frac{2 s_{i-1}}{|\Delta_{i}|}\right)^{s_i} \beta^{- s_{i-1}}   \leq 1.
\label{eq:hh}
\ee
Since $\sigma_i \geq 0$ for all $i$ it suffices to show
\be
e^{\sigma_t} (3\cdot nt/\delta) (2s_{i-1}/\Delta_{min})^{s_i} \cdot  \beta^{-s_{i-1}} \leq 1.
\label{eq:86}
\ee
At each recursion of the algorithm given $s_{i-1}$ we choose
\be
s_i = \frac{\ln \beta}{2 } \frac{s_{i-1}}{\ln \frac{2s_{i-1}}{\Delta_{min}}}.
\label{eq:de}
\ee
plugging this at the left hand side of equation \ref{eq:86} we get
\ba
e^{\sigma_t} (3\cdot nt/\delta) (2s_{i-1}/\Delta_{min})^{s_i} \beta^{-s_{i-1}} &\leq
e^{\sigma_t} (3\cdot nt/\delta) \beta^{-s_{i-1}/2},\\
&\leq e^{\sigma_t} (4nt/\delta) e^{-s_t \ln(\beta)/2}.
\ea

Consider the expression above: Since by our assumptions $\sigma_t = O(\ln(n)), \delta = 1/\poly(n), t = O(\ln(n)), \beta = O(1)$ then
to establish Equation \ref{eq:delta}
it is sufficient to show that whenever $s_0 =m$ 
satisfies the lower bound of the claim
the solution of the recursion in equation \ref{eq:de} implies 
\be
s_t = \Omega(\ln(n)).
\ee
So now, we recall the assumption in the claim that $m$ is at most poly-logarithmic and $\Delta_{min}$ is at least inverse poly-logarithmic. Hence
\be
\frac{2s_i}{\Delta_{min}}  \leq \frac{2m}{\Delta_{min}} = \poly\log(n).
\ee
Thus
\be
\ln \left(\frac{s_{i-1}}{\Delta_{min}}\right) \cdot \frac{\ln \beta}{2}
< c_0 \ln \ln(n).
\ee
for some constant $c_0 > 0$.
Hence, whenever
\be
m = \ln(n) \cdot (c_0' \ln \ln(n))^t,
\ee
for some other constant $c_0'>c_0>0$
we have that
\be
s_t = \Omega(\ln(n)).
\ee
Finally, since we assume that $t = O(\ln\ln(n)/\ln\ln\ln(n))$ it follows that
\be
m = \ln^{O(1)}(n)
\ee
thereby achieving consistency with claim assumptions.

\end{proof}



\section{Approximation of permanents of random matrices of vanishing mean}
\label{section:algo}

\subsection{Main theorem}

\begin{theorem}\label{thm:main}
Let $\delta = 1/\poly(n)$ and $\eps \geq [\ln\ln\ln(n)/\ln\ln(n)]^{1/7}$.
There exists $\mu\in [\eps, \eps + 2 \eps^2]$
such that for any distribution ${\cal D}(\mu,1)^{n\times n}$
there exists an algorithm running in quasi-polynomial time that computes a number ${\cal O}$ that for $1-o(1)$ fraction of matrices $A \sim {\cal D}(\mu,1)^{n\times n}$ satisfies
\be
{\cal O} =  \Per(A) \cdot (1 + {\cal E}), \ \, {\cal E}\in \C,\ \  |{\cal E}| = O(\delta).
\ee
In particular the algorithm solves $\mathrm{GPE}^{\eps}_{\times}$.
\end{theorem}

\begin{proof}
Let
\be
A' \sim {\cal D}(0,1)^{n\times n}.
\ee
Set $w = \pi \eps^6$.  
By Lemma \ref{lem:root2} there exists a tube ${\cal T}(\gamma_{a,b},w)$
for $|a| \leq 2 \eps, b\in [1/\eps,1/\eps + 2 \eps]$ 
that is root-free w.p. $g_{A'}(z)$
at least $1 - 4\eps$ over choices of $A'$.
We will now use the curve $\gamma_{a,b}$ to interpolate the value of the function $g(z) = \Per(J + z A')$
from $z=0$ to $z=b \in \R$.

Suppose $A'$ is such a matrix.
Divide each of the two linear segments comprising $\gamma_{a,b}$ into small equal segments of size $\Delta_{min} = w/e =  \pi \eps^6 /e$ each. 
Enumerate these segments together as $\Delta_1,\hdots, \Delta_t$ 
Then the number of segments $t$ is at most the length of $\gamma_{a,b}$ divided by $\Delta_{min}$.
By Lemma \ref{lem:root2} the length of $\gamma_{a,b}$ is at most $2/\eps$ and so 
\be
t = |\gamma_{a,b}| / \Delta_{min} \leq 2 e / (\eps w) \leq  \frac {2 e} \pi \eps^{-7} \leq e \eps^{-7}
\ee
We run algorithm \ref{algorithm:main} for the following parameters:
\begin{enumerate}
\item
matrix $A'$
\item
$\Delta_1,\hdots, \Delta_t$.  Note that $y_t = \sum_{i=1}^t \Delta_i = b$ by definition of $\gamma_{a,b}$.
\item $\beta = e$.
\item
precision parameter $\delta$
\item
$m = \ln(n) \cdot  (c_0 \cdot \ln \ln(n))^t$,
where $c_0$ is the constant implied by Claim \ref{cl:alg}
\end{enumerate}

Recall the conditions of Claim \ref{cl:alg}:
\be
\delta = n^{-\Omega(1)}, \sigma_t = O(\log(n)), 
\Delta_{min} = 1/\poly\log(n),
\ee
\be
 t = O(\ln\ln (n)/\ln\ln\ln(n)),\beta = O(1),m = \poly\log(n).
\ee

We now verify the
conditions of Claim \ref{cl:alg} in the order they appear:
\begin{enumerate}
\item
$\delta = n^{-\Omega(1)}$ by assumption. 
\item
The total length of all segments is at most $2/\eps$
by Lemma \ref{lem:root2}.
\item
$\Delta_{min} = \frac \pi e \eps^6 = \Theta (\frac{\ln \ln \ln n}{\ln \ln n})^{6/7} > 1/ \poly\log(n)$.
\item
$t \leq e \eps^{-7} = O( \frac{\ln \ln n}{\ln \ln \ln n} )$. 
\item
For each $i$ the ratio of the distance of the closest root of $g(z)$ to $y_i$ to step size $|\Delta_{i+1}|$ is at least $\beta \geq e$:
this follows by construction since ${\cal T}$ is root-free with parameter $w$, and $\Delta_{min} = w/e$.
\item
The above value of $t$ implies that
\ba
m &= 
\ln(n) \cdot (c_0 \ln\ln(n))^t\\
&= 
\ln(n) \cdot  \ln\ln (n)^{O(\ln\ln(n)/\ln\ln\ln(n))}\\
&= \poly\log(n),
\ea
\end{enumerate}
Hence we can invoke Claim \ref{cl:alg}.
By this claim for $|{\cal E}| \leq \delta$ we have:
\be
e^{\hat{f}} &= 
 \Per(J + y_t \cdot A') \cdot (1 + {\cal E})\\
 &= 
\Per(J + b \cdot A') \cdot (1 + {\cal E}).
\ee
The matrix $J + b \cdot A'$ is distributed as $b \cdot A$ where  $A \sim {\cal D}(\mu,1)^{n\times n}$
for some $\mu\in [\eps, \eps + 2 \eps^2]$.
Thus
\be
b^{-n} e^{\hat{f}} = \Per(A) \cdot (1 + {\cal E}), \quad |{\cal E}| \leq \delta
\ee
where $A\sim {\cal D}(\mu,1)^{n\times n}$.


\paragraph{Run-time:}
Algorithm \ref{algorithm:main} requires an oracle $\mathcal{O}_g$ at step \ref{it:it1} for computing derivatives of $g$ at $z=y_0=0$:  Using item (2) of Lemma \ref{lem:bar} we can implement $\mathcal{O}_g$ in $n^{O(m)}$ time.
Next, to compute the $m$ derivatives of $f(z) = \ln(g(z))$ at $z=y_0$ at step \ref{it:it2} using $\mathcal{O}_g$ we invoke item (1) of the lemma, and compute them in time at most $O(m^2 n^{O(m)}) = n^{O(m)}$.
Finally, we update at most $m$ derivatives along $t = O(\ln(n))$ steps, where each update requires at most $m$ summands.
This results in total complexity 
\be
O(t \cdot m^2) \cdot n^{O(m)} = n^{O(m)} = 2^{\poly\log(n)}.
\ee

\end{proof}

\subsection{Natural biased distributions}

The following corollaries immediately follow from this theorem by choosing $\eps = 1/\poly\log\log(n)$:
\begin{corollary}

\noindent
\begin{enumerate}
\item
\textbf{Biased Gaussian:}
There exists $\mu = 1/\poly\log\log(n), \delta = 1/\poly(n)$ and a deterministic algorithm that for $A\sim {\cal N}(\mu,1,\C)^{n\times n}$
computes a $1 + \delta$ multiplicative approximation of $\Per(A)$ on a $1-o(1)$ fraction of matrices $A$.
\item
\textbf{"Slightly-biased" Bernoulli:}
There exists $\mu = 1/\poly\log\log(n), \delta = 1/\poly(n)$ and a deterministic algorithm that for $A\sim {\rm Bern}(\mu)^{n\times n}$
computes a $1 + \delta$ multiplicative approximation of $\Per(A)$ on a $1-o(1)$ fraction of matrices $A$.

\end{enumerate}
\end{corollary}

\begin{remark}
\label{remark:subexp}
Using a tighter analysis of Algorithm \ref{algorithm:main}, one can see that, for constant $\beta$, the following is an upper bound on the number of derivatives needed
\be
m = O\Big(\ln \frac{1}{|\Delta_{min}|} + t \ln t\Big)^t
\label{eq:error-general}
\ee
We did not include a detailed proof here, since it is basically in the same spirit as the analysis given in the proof of Claim \ref{cl:alg}. The proof of Theorem \ref{thm:main} suggests that for the ensemble with $\eps$ mean, we need to choose $t = O(\frac 1 {\eps^7})$ and $|\Delta_{\text{min}}| = O(\eps^6)$. As a result, the upper bound we use for the number of derivatives as a function of the mean value $\eps$ behaves like
\be
m = \exp {O(\frac {\ln \frac 1 \eps}{\eps^7})},
\ee
and therefore, using these parameters and for $\delta = \frac 1 {\poly(n)}$, the running time of Algorithm \ref{algorithm:main} is bounded from above by $\exp ( \ln n \cdot \exp O(\frac {\ln \frac 1 \eps}{\eps^7}))$. Choosing $\eps = (\frac 1 {\ln n})^{1/9}$, we get the upper bound $\exp \exp O(\ln^{0.99} n)$ which grows strictly slower than $\exp(n^\theta)$, for any constant $\theta$.
\end{remark}

\section*{Acknowledgement}
We thank Scott Aaronson, Alexander Barvinok, and Aram Harrow for helpful discussions. LE and SM were partially supported by NSF grant CCF-1629809.

\bibliographystyle{hyperabbrv}

\appendix

\section{Inverse polynomial pean is as hard as zero mean}
\label{app:bs}

\textbf{Inverse polynomial pean is as hard as zero mean}--Consider the problem of computing the permanent to additive error defined in \cite{AA13}:

\begin{definition}  

\textbf{Additive approximation of the permanent}

\noindent
$\mathrm{GPE}^\mu_\pm (\eps, \delta)$: Given $\delta, \eps > 0$ and $A \sim \mathcal{N}(\mu,1,\C)^{n\times n}$, output a number $r$ such that with probability at least $1- \delta$
\be
\Big| \Per (A) - r\Big| \leq \sqrt{n!} \cdot \eps.
\ee
\end{definition}

\noindent
In \cite{AA13} it was shown that a hardness assumption on this problem implies that quantum computers
can sample from certain distributions that cannot be efficiently sampled by classical computers, thereby establishing a
so-called ``quantum supremacy''.
In this section we show that $\mathrm{GPE}_\pm$ and $\mathrm{GPE}_\times$ are essentially the same
for Gaussian matrix $A$ with mean $0$ and with mean $1/\poly(n)$:
\begin{theorem}

\textbf{Inverse polynomial mean is as hard as 0 mean}

\noindent
For any $K > 0$ and $\mu \leq \frac{1}{\sqrt{n-1}}$:

\begin{enumerate}
\item $\mathrm{GPE}^0_\pm(\eps + \mu \cdot n K, \delta + \frac{1}{K^2}) \preceq \mathrm{GPE}^\mu_\pm (\eps, \delta)$.

\item 
Assume that the following conjecture on anti-concentration of the permanent holds for the standard Gaussian matrix 
$A\sim {\cal N}(0,1)^{n\times n}$:
\be
\Pr_A \Big( |\Per (A)| > \frac {\sqrt{n!}}{n^c} \Big) \geq 1-\delta.
\ee
Then 
\be
\mathrm{GPE}^0_\times \Big((\eps + \mu \cdot n K) n^{c}, \delta + \delta' + \frac{1}{K^2}\Big)  \preceq \mathrm{GPE}^\mu_\times(\eps, \delta),
\ee
where $\mathrm{GPE}^0_\times(\eps, \delta)$ is defined in definition \ref{GPE1}.

\end{enumerate}
\label{thm:gpe}
\end {theorem}

\noindent
To prove this we need the following lemma:

\begin{lemma}
Let $A \sim \mathcal{D}^{n \times n} (0,1)$. 
For all real $\mu$ such that $|\mu| < \frac{1}{\sqrt{n-1}}$ and any $K > 0$ we have:
\be
\Pr_A \left(
\Big| \Per (A + J \mu) -\Per(A)\Big| <  K \cdot n\cdot \sqrt{n!} \mu
\right)
\geq
1 - K^{-2}.
\ee
\label{lem:sqrtn}
\end{lemma}

\begin{proof}
The proof is by Markov inequality inequality. Let $\Delta:= \Big | \Per (A + J \mu) -\Per(A)\Big |$. Then Markov inequality is
\be
\Pr_A \Big(\Delta \geq K \sqrt{\E_A \left[\Delta^2\right]}\Big) \leq \frac{1}{K^2}.
\ee
Hence, it is sufficient to show that if $\mu < \frac{1}{\sqrt{n+1}}$ then 
\be
\E_A \Big[\Delta^2\Big] \leq  n! n^2 \mu^2.
\ee
To see this we write $\Per (A + \mu J)$ as a degree-$n$ polynomial in $\mu$: 
\be
\Per(A + \mu J) = 
\sum_{k=0}^n a_k \cdot \mu^k
\ee
Under this notation we get:
\be
\Delta = \Big| \sum_{k=0}^n a_k \cdot \mu^k - a_0 \mu^0\Big| = 
\Big| \sum_{k=1}^n a_k \cdot \mu^k \Big|
\ee
Therefore 
\be
\E_A \left[\Delta^2\right] = \sum_{i, j = 1 }^n \E_A a^\ast_i a_j \cdot \mu_i \mu_j
\ee
Using the pairwise independence of $a_i$ and $a_j$ (see the proof of lemma \ref{lemma:moment}), i.e., $\E_A a^\ast_i a_j \propto \delta_{i,j}$, we have that $\E_A \Delta^2 = \sum_{i = 1 }^n \E_A |a_i|^2 \cdot \mu^{2i}$. 

Again, using the observation in the proof of lemma \ref{lemma:moment}, by pair-wise independence of the permanent of different sub-matrices, i.e., 
$\E_A (\Per(C) \Per(B)) = 0$ when $C$ and $B$ are different sub-matrices:
\be
a_k = k! \sum_{B \subseteq_{n-k} A} \Per(B)
\ee
\be
 \Downarrow \nonumber
 \ee
\be
 \E_A |a_k|^2 = k!^2 \sum_{B \subseteq_{n-k} A }|\Per(B)|^2 = \frac{n!^2}{(n-k)!}.
\ee
Summing over these terms 
\ba
\E_A \left[\Delta^2\right] &= n!^2  \sum_{k=1}^n \frac{\mu^{2k}}{(n-k)!}\\
&=n!^2  \sum_{k=1}^n c_k,\\
\ea
and
\ba
 c_k:= \mu^{2k} / (n-k)!.
\ea
Finally we observe that if $\mu < \frac{1}{\sqrt{n-1}}$ the largest term in the sum above corresponds to $k=1$. To see this note for $ 1\leq k \leq n$
\be
\frac{c_{k+1}}{c_k} = \mu^2 (n-k)\leq \frac{n-k}{n-1}\leq 1.
\ee
Therefore
$\E_A \left[\Delta^2\right] \leq (n!)^2 \cdot n \cdot c_1 = n! n^2\mu^2$.
This completes the proof.

\end{proof}
\begin{proof}(of theorem \ref{thm:gpe}.) For the first part suppose that we have an oracle $\mathcal{O}$ that solves $\mathrm{GPE}^\mu_\pm (\eps, \delta)$. Then given an instance $A$ from $\mathrm{GPE}^0_\pm$ we just return the output $r$ of $\mathcal{O}$. Therefore using union bound, triangle inequality and Lemma \ref{lem:sqrtn} with probability at least $1- \delta - 1/K^2$
\ba
\Big |\Per(A) - r\Big| &\leq \Big |\Per(A) - \Per (A + \mu \cdot J)\Big| + \eps \sqrt{n!}\\
&\leq (\eps + K \cdot n\cdot \mu) \sqrt{n!}.
\ea

For the second part we consider the same reduction, and the proof is immediate from the first item by just observing that if with probability $1-\delta'$, $\sqrt{n!} \leq n^c \cdot |\Per (A)|$ then using union bound with probability at least $1- \delta -\delta'- 1/K^2$
\ba
\Big |\Per(A) - r\Big| &\leq 
(\eps + K \cdot n\cdot  \mu) \sqrt{n!}\\
&\leq 
(\eps + K \cdot n\cdot  \mu) n^c |\Per(A)|.
\ea
\end{proof}

\section{Proof of Lemma \ref{lem:int1}}
\label{section:lem20}

\textbf{Proof of Lemma \ref{lem:int1}}--
\begin{proof}
Let $c:= \frac{e^{\frac l m}}{\beta} \leq \frac{1}{\sqrt e}$:
\ba
\sum_{k = m}^\infty \beta^{-k} \cdot k^l &= \beta^{-m} \cdot m^l \sum_{k = m}^\infty \beta^{-(k-m)} \cdot (\frac k m)^l\\
&= \beta^{-m} \cdot m^l \sum_{k = m}^\infty \beta^{-(k-m)} \cdot (1 + \frac {k-m} m)^l\\
&= \beta^{-m} \cdot m^l \sum_{k = 0}^\infty \beta^{- k } \cdot (1 + \frac {k} m)^l\\
&\leq \beta^{-m} \cdot m^l \sum_{k = 0}^\infty \beta^{- k } \cdot e^{ \frac {k} m l}\\
&\leq \beta^{-m} \cdot m^l \sum_{k = 0}^\infty c^k\\
&= \beta^{-m} \cdot m^l \frac{1}{1-c}\\
&\leq 3\cdot \beta^{-m} \cdot m^l.
\ea

\end{proof}

\section{Average $\# P$-hardness for the exact computation of the permanent of a random Gaussian with vanishing mean}
\label{section:averagevanishing}
\textbf{Average $\# P$-hardness for the exact computation of the permanent of a random Gaussian with vanishing mean}--
Our result implies a quasi-polynomial time algorithm to approximate the permanent of a Gaussian matrix with vanishing but nonzero mean and variance $1$. One may wonder if a stronger statement can hold: ``Is there a quasi-polynomial time algorithm to compute the permanent of such matrix exactly on average''. In this appendix we prove that the answer is no, unless $P^{\# P} \subseteq \mathrm {TIME}(2^{\poly\log (n)})$. Our result therefore provides a natural example of a computational problem that is $\#P$-hard to compute exactly on average and efficient to approximate on average.

\begin{theorem}
[Average-case hardness for the permanent of a nonzero mean Gaussian] Let $\mu > 0$ and let $\mathcal{O}$ be the oracle that for any $\mu' \geq \mu$ computes the permanent of $7/8 +1/\poly(n)$ fraction of matrices from the ensemble ${\cal N}^{n \times n}(\mu', 1, \C)$ exactly. Then $P^{\mathcal{O}} = P^{\# P}$.
\label{thm:exacthardness}
\end{theorem}

What this theorem is saying is that using $\poly(n)$ queries to $\mathcal{O}$ and polynomial time computation one can compute the permanent for $3/4 + 1/\poly(n)$ fraction of $A \sim \mathcal {N}^{n\times n} (0 , 1, \C)$ for any $\mu'' < \mu$, including $\mu'' = 0$. In other words efficient computation of the permanent of a matrix with mean larger than a certain amount with high probability implies efficient computation of the permanent of random matrix with arbitrary mean with high probability.

\begin{table}
\captionsetup{font=scriptsize}
\captionsetup{width=13cm}
\centering
\captionsetup{width=.8\linewidth}
\begin{tabular}{ |c|  c|  c| }
\hline
   & worst case & average case \\
   \hline
 exact & $\#P$-hard & $\#P$-hard \\
 \hline
   approximate & $\#P$-hard & Efficient \\
   \hline
\end{tabular}
\caption{The computational complexity of computing the permanent of a complex Gaussian matrix with nonzero but vanishing mean.}
\label{table:2}
\end{table}

In order to show this we use the following result of Aaronson and Arkhipov \cite{AA13}
\begin{theorem}
[Aaronson Arkhipov \cite{AA13}] It is $\#P$-hard to compute the permanent of a Gaussian matrix with mean $0$ and variance $1$ for $3/4 + 1/\poly(n)$ fraction of matrices.
\label{thm:AAaverage}
\end{theorem}
and the following algorithm due to Berlekamp-Welch:
\begin{theorem} [Berlekamp-Welch Algorithm]
Let $q: \C \rightarrow \C$ be a univariate polynomial of degree $n$. Suppose we are given $m$ pairs $(x_1, y_1),...,(x_m, y_m)$ (with the
$x_i$'s all distinct), and are promised that $y_i = q (x_i)$ for more than $\frac{m+n} 2$
values of $i$. Then there is
a deterministic algorithm to reconstruct $q$, using $\poly (m,n)$ operations.
\label{thm:bwalgorithm}
\end{theorem}

\begin{proof} [Proof of Theorem \ref{thm:exacthardness}]
The ensemble of the permanent of random Gaussian matrices with mean $\mu$ is according to
\be
\Per(J \cdot \mu + A), \hspace{1cm} A \sim \mathcal{N}^{n\times n} (0,1, \C).
\ee

Consider the following univariate polynomial in $\mu$
\be
q(\mu):= \Per(J \cdot \mu + A) = \sum_i c_i \mu^i.
\ee
Our objective is to compute
\be c_0 = \Per (A).\ee

Suppose that there is an oracle $\mathcal{O}$ that for any $\mu' > \mu$ computes $\Per (A)$ with probability at least $7/8 + 1/\poly(n)$ over $A \sim \mathcal{N}^{n \times n}(\mu', \C)$ exactly. Let $A \sim \mathcal{N}^{n \times n}(0,1, \C)$. Fix distinct valuess $\mu_1, \ldots, \mu_m > \mu$ for $m = \poly(n)$, and query the oracle on $A + J \mu_i$ for $i \in [m]$.
 Let
\be
N_{\text{wrong}} = \sum_{i=1}^m I(q(\mu_i) \neq \mathcal O(\mu_i)),
\ee
where $I (E)$ is the indicator of the event $E$.
We know that
\be
\E_A N_{\text{wrong}} \leq m/8 (1+1/\poly(n)).
\ee
Let $m = n^{2 + O(1)}$. 

Hence
\ba
\Pr_A (N_{\text{wrong}} \geq \frac{m - n}{2})
&\leq \frac{ m}{4(m - n)} (1+1/\poly(n))\\
&= \frac 14( 1+ 1/\poly(n)).
\ea
Using Theorem \ref{thm:bwalgorithm}, conditioned on the event $N_{\text{wrong}} < \frac{m - n}{2}$ there is a polynomial time algorithm that reconstructs $q$ and hence can find $\Per(A)$. This implies a procedure to compute $\Per(A)$ exactly for a fraction 
at least
$3/4+ 1/\poly(n)$
of such matrices. Using Theorem \ref{thm:AAaverage} we conclude that $P^{\mathcal O} = P^{\#P}$.

\end{proof}

\end{document}